\documentclass{scrartcl}
\PassOptionsToPackage{arxiv,full}{optional}

\usepackage[english]{babel}
\usepackage{csquotes}
\usepackage{booktabs}
\usepackage{optional}

\opt{aamas}{
	\usepackage{amsmath,amssymb,amsthm}
	\usepackage{nicefrac}
	\usepackage{algorithm}
	\usepackage[noend]{algpseudocode}
	\usepackage{enumitem}
}

\opt{arxiv}{
	\usepackage[noblocks]{authblk}
	\usepackage{amsmath,amssymb,amsthm}
	\usepackage{nicefrac}
	\usepackage{algorithm}
	\usepackage[noend]{algpseudocode}
	\usepackage{enumitem}
	\usepackage{authblk}
	\usepackage{hyperref}
	\usepackage[natbib,maxbibnames=99,style=alphabetic]{biblatex}
	\addbibresource{bibliography_clean.bib}
}
\opt{jaamas}{
	\smartqed  %
	\usepackage{graphicx}
	\usepackage{amsmath,amssymb}
	\usepackage{nicefrac}
	\usepackage{algorithm}
	\usepackage[noend]{algpseudocode}
	\usepackage{enumitem}
	\usepackage[numbers]{natbib}
}

\newcommand{\fenhg}{FEN-hedonic game}
\algrenewcommand\algorithmicrequire{\textbf{Requires:}}

\opt{arxiv}{
	\newtheorem{theorem}{Theorem}
	
	\newtheorem{lemma}[theorem]{Lemma}
	
	\newtheorem{definition}[theorem]{Definition}
	
}
\newtheorem{observation}[theorem]{Observation}

\usepackage[dvipsnames]{xcolor}
\usepackage{tikz}
\newcommand{\reva}[1]{#1}
\newcommand{\revb}[1]{#1}
\newcommand{\revc}[1]{#1}
\newcommand{\revm}[1]{#1}
\newcommand{\revnew}[1]{#1}

\begin{document}
\opt{conference}{
	\title{Testing In\-di\-vi\-du\-al-Based Stability Properties\\ in Graphical Hedonic Games}
}
\opt{full}{
	\title{Testing Stability Properties in Graphical Hedonic Games}
}
\opt{aamas}{
 	\author{Hendrik Fichtenberger}
 	\authornote{This author acknowledges the support by ERC grant No. 307696.}
 	\orcid{https://orcid.org/0000-0003-3246-5323}
 	\affiliation{TU Dortmund, Germany}
 	\email{hendrik.fichtenberger@tu-dortmund.de}
 	\author{Amer Krivošija}
 	\affiliation{TU Dortmund, Germany}
 	\email{amer.krivosija@tu-dortmund.de}
 	\author{Anja Rey}
 	\affiliation{TU Dortmund, Germany}
 	\email{anja.rey@tu-dortmund.de}
}
\opt{arxiv}{
	\opt{conference}{
		\author{Hendrik Fichtenberger\thanks{This author acknowledges the support by ERC grant No. 307696.} }
		\affil{TU Dortmund, Germany\\ \href{mailto:hendrik.fichtenberger@tu-dortmund.de}{hendrik.fichtenberger@tu-dortmund.de}\\ \href{https://orcid.org/0000-0003-3246-5323}{https://orcid.org/0000-0003-3246-5323}}
		\author{Amer Krivošija}
		\affil{TU Dortmund, Germany\\
		\href{mailto:amer.krivosija@tu-dortmund.de}{amer.krivosija@tu-dortmund.de}}
		\author{Anja Rey}
		\affil{TU Dortmund, Germany\\ \href{mailto:anja.rey@tu-dortmund.de}{anja.rey@tu-dortmund.de} \href{https://orcid.org/0000-0001-7387-6853}{https://orcid.org/0000-0001-7387-6853}}
	}
	\opt{full}{
		\author{Hendrik Fichtenberger}
		\affil{University of Vienna, Austria\\ \href{mailto:hendrik.fichtenberger@univie.ac.at}{hendrik.fichtenberger@univie.ac.at}\\ \href{https://orcid.org/0000-0003-3246-5323}{https://orcid.org/0000-0003-3246-5323}}
		\author{Anja Rey}
		\affil{University of Cologne, Germany\\ \href{mailto:rey@cs.uni-koeln.de}{rey@cs.uni-koeln.de}\\ \href{https://orcid.org/0000-0001-7387-6853}{https://orcid.org/0000-0001-7387-6853}}
	}
	\date{}
}
\opt{jaamas}{

	\author{Hendrik Fichtenberger \and
		Anja Rey %
	}

	\institute{
          H. Fichtenberger \at
          University of Vienna, Austria.
              \email{hendrik.fichtenberger@univie.ac.at}
		  \and
		  A. Rey \at
		      University of Cologne, Germany.
		      \email{rey@cs.uni-koeln.de}
	}

	\date{Received: date / Accepted: date}

	\maketitle
	
}

\opt{arxiv}{
	\maketitle
}

\begin{abstract}
In hedonic games, players form coalitions based on individual preferences over the group of players they could belong to. Several concepts to describe the stability of coalition structures in a game have been proposed and analysed in the literature. However, prior research focuses on algorithms with time complexity that is at least linear in the input size. In the light of very large games that arise from, e.g., social networks and advertising, we initiate the study of sublinear time property testing algorithms for existence and verification problems under several notions of coalition stability in a model of hedonic games represented by graphs with bounded degree. In graph property testing, one shall decide whether a given input has a property (e.g., a game admits a stable coalition structure) or is far from it, i.e., one has to modify at least an $\epsilon$-fraction of the input (e.g., the game's preferences) to make it have the property. In particular, we consider verification of perfection, individual rationality, Nash stability,
\opt{conference}{%
  and (contractual) individual stability.
}%
\opt{full}{%
  (contractual) individual stability, and core stability.
}%
While there is always a Nash-stable coalition structure (which also implies individually stable coalitions), we show that the existence of a perfect coalition structure \revc{is not tautological but} can be tested. All our testers have one-sided error and time complexity that is independent of the input size.
\opt{jaamas}{%
	\footnote{\revm{A preliminary version of this work was published in the proceedings of AAMAS~2019 (International Conference on Autonomous Agents and Multiagent Systems).}}
}
\end{abstract}

\opt{aamas}{
	\keywords{cooperative games; sublinear algorithms; hedonic games; stability; property testing}
	\maketitle
}

\section{Introduction}

Hedonic games are a form of coalition formation games, in which players form teams in a decentralized manner based on individual preferences over coalitions, i.e., subsets of players. The solution to such a game is a coalition structure, i.e., a partition of the set of players.
The main idea of hedonic games is that the players' evaluation of a coalition structure only depends on their own coalitions and not on how other players work together~\cite{dre-gre:j:hedonic-coalitions}.
These games have been formalized by \citet{ban-kon-soen:j:core-simple-coalition-formation-game} and \citet{bog-jac:j:stability-hedonic-coalition-structures}, independently.
In order to evaluate the quality of a coalition structure, several solution concepts have been considered. These include, e.g., Nash stability, which states that no individual player wants to deviate from the coalition structure, and core stability, which requires that no group of players wants to deviate and form a coalition of their own.

Algorithmically, a key issue is to find suitable representations of hedonic games: Since the number of possible coalitions for a player is exponential in the number of players, there is a trade-off between compactness and expressivity of the preference profile.
In the areas of \emph{Cooperative Games in Multiagent Systems}~(see~\citep{cha-elk-woo:b:ccgt} for details) and \emph{Computational Social Choice}~(see~\citep{hg-handbook} for details), a number of representations and stability concepts are analysed with respect to the computational complexity of deciding whether there exists a stable solution, verifying whether a given solution is stable, and finding a stable solution.
Even for restricted representations such as additively separable games, these questions are often intractable.
For instance, it is often $\mathrm{NP}$-complete to decide whether a given game allows a Nash-stable coalition structure, see~\citet{pet:graphical-hg-bounded-tw:16} for an example.
The existence of core-stability is often even $\Sigma_2^p$-complete to decide; examples can be found in~\citet{woe:ashg-sigma2p-hardness-core:13} and~\citet{ota-etal:hg-fen:17}.
This strikes even harder when the considered game instances are very large because they arise from, e.g., social networks or the assignment of advertisements to available slots on web pages so that competing ads do not interfere \revc{(this can be configured in real advertisement services \cite{googleads1,googleads2})}.
Here, it might already be impractical to read the whole input once because the data does not fit into memory or the access is slow or restricted.

In this \opt{conference}{paper}\opt{full}{article}, we study sublinear algorithms for hedonic games.
We aim to decide in sublinear time whether a game has a stable coalition structure or is far from this with respect to the number of required changes of preferences such that it admits a stable coalition structure, as well as whether a game is stable under a given coalition structure $\Gamma$ or is far from being stable under $\Gamma$.
When a coalition can be stabilized by only few compromises on the preferences, it may be acceptable to sustain the situation.
When, however, too many modifications are required to obtain any stable situation, the current situation is too far off.

Graph representations provide a compact means to encode structural connections between players. 
A formal study of graphical hedonic games is provided by \citet{pet:graphical-hg-bounded-tw:16}.
A popular variant is to encode a game as a network where players correspond to vertices and edges illustrate friendship relations. Players that are not friends are often referred to as enemies.
Preferences are extended either by prioritising appreciation of friends or aversion to enemies~\cite{dim-etal:hg-fe:06}.
However, if the game is very large, many players may not be involved in any relationship.
In this scenario, it is natural to consider a more general model.
For each player, the set of other players is divided into three subsets: friends, enemies and neutral players~\cite{ota-etal:hg-fen:17}, which is what we call the FEN-encoding.
Furthermore, we bound the number of friends and enemies per player by a constant.
\revc{For example, if the players are humans, it is known that each player can only maintain a limited number of stable relationships to other humans~\cite{DUNBAR1992469,dunbar1998grooming}. This number is known as \emph{Dunbar's number}, and empirical estimates are smaller than 300~\cite{hernando2010unravelling}.
If the players are advertisements, advertisers might expect positive or negative effects for a limited number of other advertisements when placed on the same web page (e.g., an ad for a phone and an ad for an accessory item might up-sell while two ads for comparable phone models might only cross-sell).}
Under restrictions such as bounded degree and bounded treewidth, some stability questions become solvable in linear time~\cite{pet:graphical-hg-bounded-tw:16}.
Nevertheless, this still incurs the evaluation of the whole game in order to verify whether a coalition structure is stable.
Given the local views of individual vertices within hedonic games, it would be preferable and much more practical to ask only a sample of players for their individual preferences and deduce global properties.

The area of property testing provides a framework to relax such decision problems in favour of sublinear complexity (see~\citep{gol:b:pt:2017} for an overview).
A property tester is a randomized algorithm that decides, with error probability at most $\nicefrac{1}{3}$, whether the input satisfies some property~$\mathcal{P}$ or is far from satisfying $\mathcal{P}$ by probing only a small part of it.
In the setting of graph properties, a graph~$G$ with bounded vertex degree $d$ is $\epsilon$-far from satisfying some property~$\mathcal{P}$ (e.g., bipartiteness) if one has to modify at least $\epsilon d n$ edges to make $G$ have property~$\mathcal{P}$.
If the property tester always accepts graphs in $\mathcal{P}$, it has one-sided error; otherwise, it has two-sided error.
The input graph~$G$ may be probed by the algorithm through an oracle that provides access to the entries of the adjacency lists of $G$, and the computational complexity of the property tester is measured in terms of queries it asks.

In comparison to classic decision problems, property testing problems allow for algorithms with sublinear complexity. For example, a randomized decision algorithm for graph connectivity needs to read the whole input to achieve constant error probability, which implies a linear lower bound on the complexity. In contrast, a property tester for connectivity has only constant complexity~\cite{GolPro02}. This difference arises because the property tester does not need to read the whole input, and, in fact, sublinear complexity renders this impossible. Therefore, the input model plays an important role in property testing. While there is a characterization for constant query testable properties in dense graphs (graphs with $\Omega(\lvert V \rvert^2)$ edges) \cite{AloCom09}, less is known for graphs with bounded degree and general graphs.

\subsection{Our Contribution}
We study property testing of stability problems in \fenhg s, where each player has a bounded number of symmetric relationships to friends and enemies as represented by labelled edges of an undirected graph, and preferences are extended to coalitions by any utility function linear in the number of friends and enemies in a coalition.
The setting of hedonic games enhances graphs by rich semantics, which stands in contrast to purely combinatorial and geometric properties previously studied in graph property testing.
We model the semantics of hedonic games as an additional layer on top of the combinatorial graph structure and analyse existence and verification problems for various stability concepts.
In particular, we study
common
individual-based
stability concepts
such as perfection, individual rationality, Nash stability, and (contractual) individual stability%
\opt{full}{
  as well as core stability%
}%
.

While individually rational, Nash-stable, individually stable, and contractually individually stable coalitions always exist,
there are games which do not allow a perfect coalition structure.

\begin{theorem}\label{mainthm:existence}
	Given an \fenhg~$G$ with bounded degree $d$, it can be tested whether $G$ admits a perfect coalition structure with bounded coalition size~$c$. The tester has one-sided error and query complexity $\mathrm{poly}(\epsilon, c, d)$.
\end{theorem}
\revm{We provide the technical result as Theorem~\ref{thm:perfection-bounded}
in Section~\ref{sec:existence}.}%

While the existence problem as to whether a game allows a stable outcome is a property of edge-labelled graphs, the verification problem of whether a game satisfies stability according to a given coalition structure~$\Gamma$ requires additional modelling: We assume that next to oracle access to the adjacency lists of the underlying bounded-degree graph of a game $G$, we have additional access to an oracle to~$\Gamma$, i.e., \revb{to the given} partition of the vertex set.

\begin{theorem}\label{mainthm:verification}
	Given an \fenhg~$G$ with bounded degree $d$ and a coalition structure $\Gamma$, it can be tested whether $G$ is stable under~$\Gamma$ with respect to perfection, individual rationality, Nash stability, individual stability and contractual individual stability with one-sided error and query complexity $\mathrm{poly}(\epsilon, d)$.
\opt{full}{
	For core stability we obtain a one-sided error tester with query complexity dependent on $\epsilon$, $c$, and $d$, where $c$ is the maximum coalition size.
}
\end{theorem}
\revm{We provide the technical result as Theorem~\ref{thm:verification-ind-stability} based on Theorem~\ref{thm:verification-tester} in Section~\ref{sec:verification}.}

Note that while we consider $c$ and $d$ to be of constant size, independent of the input size~$n$, our statements remain valid if, for instance, $d\in\mathcal{O}(\log n)$. We provide extensions of our theorems to weighted graphs (additively separable games) and directed graphs (asymmetric games) in Section~\ref{sec:extensions}.

\revm{A preliminary version~\cite{hgpt:c:2019} of this work was published in the proceedings of \mbox{AAMAS}~2019 (International Conference on Autonomous Agents and Multiagent Systems). This is a revised version that also extends the preliminary version by a verification tester for core stability (see Theorem~\ref{thm:verification-ind-stability}), an extension of our main theorems to weighted graphs (see Theorems~\ref{thm:weighted-verification-tester} and~\ref{thm:weighted-perfection-bounded}), an extension of our main theorem on verification testing to directed graphs (see Theorem~\ref{thm:verification-directed}) and a lower bound related to our main theorem on testing the existence of a perfect coalition structure for directed graphs (see Theorem~\ref{thm:perfection-bound-unidirectional}).}

\subsection{Related Work}

Hedonic games were formally defined by~\citet{ban-kon-soen:j:core-simple-coalition-formation-game} and \citet{bog-jac:j:stability-hedonic-coalition-structures}.
A well-known application of a restricted variant with size-two coalitions is the stable-roommates problem~\cite{stable-roommates} for the allocation of student houses.
Mostly, hedonic games
have been analysed from a computational complexity point of view with respect to a trade-off between expressivity, succinct representation and tractability of stability decision problems. The complexity of general hedonic games has first been studied by \citet{bal:hg-complexity}. The worst-case complexity of stability problems for various representations and different stability concepts has been studied extensively \revb{since}:
Popular representations include additively separable hedonic games~\cite{bog-jac:j:stability-hedonic-coalition-structures,azi-bra-see:j:ashg,woe:ashg-sigma2p-hardness-core:13}, singleton encodings~\cite{cec-rom:j:singleton-encoding}, hedonic coalition nets~\cite{elk-woo:hedonic-coalition-nets}, and dichotomous preferences~\cite{azi-har-lan-woo:c:boolean-hg}; see also~\citet{hg-handbook} and~\citet{cha-elk-woo:b:ccgt} for an overview.
\citet{pet-elk:c-coopmas:simple-causes-complexity-hg} analysed causes of and conditions for hardness.
The existence of Nash stability and other individual stability concepts is often (if not guaranteed \revc{to exist}) $\mathrm{NP}$-complete to decide (see~\citet{SunCom10} for an analysis of additively separable games). For core stability, this is often even harder, namely $\Sigma_2^p$-complete~\cite{woe:ashg-sigma2p-hardness-core:13,ota-etal:hg-fen:17}.
\citet{dim-etal:hg-fe:06} defined restricted hedonic games based on a network of friends and enemies. A more general version including neutral players was defined by \citet{ota-etal:hg-fen:17}.
Games with neutral players and partial individual evaluations were studied by \citet{lan-rey-rot-sch-sch:c:hgopt} and \citet{pet:graphical-hg-bounded-tw:16}.
In the latter work, in particular, the authors considered a constant bound on the number of individual preferences and studied graphical hedonic games with bounded treewidth.
With this restriction, it can be decided in linear time whether, for instance, a Nash-stable coalition structure exists.
A graphical model restricting to the formation of coalitions between players that are connected in the underlying graph was studied by \citet{iga-elk:c:hg-graph-restricted}.
\citet{DarGro18} studied games where players have preferences over different types of coalitions (activities) and their sizes instead of \revb{the participating} players \revb{(this can also be expressed as a hedonic game)}. For the case that there is only one type of coalition, \citet{LeeSta15} and \citet{LeeCom17} extended the setting by \citeauthor{DarGro18} so that each player may have a bounded number of friend-enemy relationships to other players.

\revb{The most related work to ours in the area of property testing is that for classic properties of graphs as there is, to the best of the authors' knowledge, no work on game theory so far. For example, \citet{GolPro02} showed that well-known properties of bounded-degree graphs like cycle freeness are testable with constant query complexity and \emph{two-sided error}. However, for \emph{one-sided error} testing of cycles freeness, $\Theta(\sqrt{n})$ queries are required \citep{GolPro02} as well as sufficient as \citet{doi:10.1002/rsa.20462} show. Property testing of annotated (or labelled) graphs has been studied for geometric graphs mainly, e.g., by \citet{ben2007lower}, \citet{czumaj2008testing} and \citet{hellweg2010testing}.}

Learning hedonic games was studied by \citet{SliLea17}. While property testing focuses on testing whether a game admits a stable outcome or whether an outcome is stable with sublinear complexity, PAC learning constructs a good hypothesis and \emph{PAC stabilization} uses this hypothesis to learn a stable outcome (if possible) using a superlinear number of samples (with possibly linear size).
As far as we know, no sublinear algorithms have been developed for hedonic games, yet.

\section{Preliminaries}

In this \opt{conference}{paper}\opt{full}{article}, we consider
graphs with vertex degrees bounded by a constant $d$.
\revc{If not stated otherwise, graphs are undirected and unweighted.}
For a graph $G=(V,E)$ at hand, we write $n = \lvert V \rvert$.
Without loss of generality, we assume that $V = [n] = \{1,\dots,n\}$.

\subsection{Hedonic Games}
A \emph{hedonic game}~$(N,\succeq)$ consists of a set of \emph{players}~$N=[n]$ and a \emph{preference profile}~$\succeq=(\succeq_1,\dots,\succeq_n)$, where $\succeq_i$ is player $i$'s preference relation over $\mathcal{N}_i=\{C\subseteq N\mid i\in C\}$. A subset $C\subseteq N$ of players is called \revb{a} \emph{coalition}.
An output of a hedonic game is a \emph{coalition structure}, i.e., a partition~$\Gamma$ of the player set. Let $\Gamma(i)\in\Gamma$ be the coalition containing~$i\in N$.
We say that player~$i$ \emph{weakly prefers} coalition~$A$ to coalition~$B$, if $A\succeq_i B$. Player~$i$ \emph{prefers} $A$ to $B$, denoted by $A\succ_i B$, if $A\succeq_i B$, but $B\not\succeq_i A$; $i$ is \emph{indifferent} between $A$ and $B$, denoted by $A\sim_i B$,
if $A\succeq_i B$ and $B\succeq_i A$.

Since the set $\mathcal{N}_i$ of coalitions a player is contained in has an exponential size in the number of players, a central question in the study of hedonic games is to define representations that are adequately compact and at the same time as expressive as possible.

One common representation is that of a graph network, where the players \reva{in $N$} are vertices in the graph. In the encoding as defined by \citet{ota-etal:hg-fen:17},
for each player $i\in N$, there exists a set $N_i^+\subseteq N\setminus\{i\}$ of \emph{friends} \revb{and a set $N_i^-\subseteq N\setminus\{i\}$} of \emph{enemies}, $N_i^+\cap N_i^-=\emptyset$.
The remaining players are considered as \emph{neutral} $N_i^0=N\setminus (N_i^+\cup N_i^-\cup\{i\})$.
We call this representation \emph{FEN-encoding}.
It can be represented by a labelled graph~$G=(N,F\cup E)$ with $F\cap E=\emptyset$,
where $j\in N_i^+$ if and only if $(i,j)\in F$, and $j\in N_i^-$ if and only if $(i,j)\in E$.
\opt{full}{\revm{%
We distinguish between the following cases:
\begin{description}
 \item[\rmfamily\mdseries \textit{undirected, unweighted FEN-encoding:}]
   If not stated otherwise, we consider undirected and unweighted graphs.
   Hence, we have symmetric preferences, i.e., $i\in N_j^+$ if and only if $j\in N_i^+$ and $i\in N_j^-$ if and only if $j\in N_i^-$.
 \item[\rmfamily\mdseries \textit{directed, unweighted FEN-encoding:}]
   If the graph is directed, preferences can be asymmetric, e.g., player~$i$ can consider $j$ as a friend, but $j$ can consider $i$ as neutral or even an enemy.
 \item[\rmfamily\mdseries \textit{weighted FEN-encoding:}]
   In addition edges can be weighted, i.e., each player~$i$ specifies a value $w(i,j)$ for each other player~$j\in N_i$. We assume that $w(i,j)$ is positive if $j\in N_i^+$, and negative if $j\in N_i^-$.
   Here edges can also be either undirected ($w(i,j)=w(j,i)$) or directed (possibly different values for $w(i,j)$ and $w(j,i)$).
\end{description}
}%
}
\revnew{The latter two cases are studied in Section~\ref{sec:extensions}.}
\revm{In the unweighted case}, we extend the players' relations to preferences in the following manner.
A value function is specified such that each player~$i\in N$ assigns
\revb{a fixed positive value to each friend and a fixed negative value to each enemy. Formally, for two values~$f,e\in\mathbb{R}_{>0}$, a}
corresponding utility function $u_i:\mathcal{N}_i\to\mathbb{R}$, $i\in N$, is defined additively by \[u_i(C)=f\cdot\vert C\cap N_i^+\vert -e\cdot\vert C\cap N_i^-\vert.\]
For instance, \emph{friends appreciation}
\revb{can be represented by} $f=d$ and $e=1$,
and \emph{enemies aversion} \revb{by} $f=1$ and $e=d$.
The preference extension is obtained by $A\succeq_i B\iff u_i(A)\geq u_i(B)$.

\begin{definition}[FEN-hedonic game]
 \label{def:fenhg}
 \revb{Let $(N,\succeq)$ be} a hedonic game represented by an \revb{undirected and unweighted} FEN-encoding, with preference profile~$\succeq$ extended via \revc{utility function~$u_i$}.
 \revb{We call $(N,\succeq)$} \emph{\fenhg}.
 
 \opt{full}{\revm{For directed and weighted FEN-encodings, we refer to such games as \emph{directed} and \emph{weighted} \fenhg s, respectively}.}
\end{definition}

\reva{Figure~\ref{fig:sc} below shows two examples of a \fenhg.}
\revc{This conforms to the definition of a graphical hedonic game~\cite{pet:graphical-hg-bounded-tw:16}
such that a player~$i$'s preference of a coalition $C\in\mathcal{N}_i$ over a coalition $D\in\mathcal{N}_i$ only depends on $i$'s neighbourhood $N_i=N_i^+\cup N_i^-$:
\begin{align}
 C\succeq_i D \iff C\cap N_i \succeq_i D\cap N_i. \label{eq:comparison}
\end{align}
}%
Note that \emph{responsiveness} is always satisfied, i.e., $C\cup\{j\}\succ_i C$ and $C\succ_i C\cup\{j'\}$, for each $i\in N$, and each $C\in\mathcal{N}_i$ and $j\in N_i^+$, $j'\in N_i^-$.

\opt{full}{%
\revm{Weighted and directed \fenhg s are equivalent to additively separable hedonic games~\cite{bog-jac:j:stability-hedonic-coalition-structures} via $w(i,j)=0$ for $j\in N_i^0$ and \[u_i(C)=\sum_{j\in C\smallsetminus\{i\}}w(i,j)=\sum_{j\in C\cap N_i}w(i,j);\] weighted, undirected \fenhg s are equivalent to symmetric additively separable hedonic games.}
}

Furthermore, we make the following assumptions.
We consider graphs of bounded degree $\vert N_i\vert \leq d$ represented by an adjacency list; in particular, it can be decided in time independent of the number~$n$ of players whether $C\succeq_i D$, and independent of the coalition size $\vert C\vert $ and $\vert D\vert $.
Moreover, it is often useful to restrict the coalition size, e.g., when players are people that have to communicate or when a coalition represents all ads displayed on a single web page. Therefore, we also consider a bounded coalition size of $\vert C\vert \leq c$.

By $\mathfrak{G}_n$ we denote the set of graphs with $n$ vertices that represent such a game.
The set of coalition structures partitioning $n$ players is denoted by $\mathfrak{C}_n$.
In the following, let $G=(N,F\cup E)$ be a graph that represents a \fenhg\ 
and let $\Gamma\in\mathfrak{C}_n$ be a coalition structure solving this game.
There are several solution concepts motivated from different perspectives on the game.
\revc{%
Note that in comparison to the definitions in the literature we take into account the bound on the coalition size $|C|\leq c$. Whenever $c=n$ holds, the stability concepts are equivalent to their original definition.
}%
Let $\mathrm{Fav}(i)$ denote the set of player~$i$'s favourite coalitions of size at most $c$, i.e., those coalitions that $i$ weakly prefers over all other coalitions~$C\in\mathcal{N}_i$ of size $\vert C\vert \leq c$.
On the one hand, $\Gamma$ is called
\begin{description}
\item[\rmfamily\mdseries \textit{perfect}] if each player~$i\in N$ weakly prefers $\Gamma(i)$ to every coalition, i.e., $\Gamma(i)\in\mathrm{Fav}(i)$.
\end{description}
This property reflects an ideal situation, but is rather rarely fulfilled.

On the other hand, $\Gamma$ is called
\begin{description}
\item[\rmfamily\mdseries \textit{individually rational}] if for each~$i\in N$,
$\Gamma(i)$ is \emph{acceptable}, i.e., $\Gamma(i)\succeq_i\{i\}$.
\end{description}
Individual rationality is guaranteed by $\{\{i\}\mid i\in N\}$.
Other stability notions are based, for example, on the lack of deviations of a single player to another (possibly empty) existing coalition.
A coalition structure $\Gamma$ is called
\begin{description}
\item[\rmfamily\mdseries \textit{Nash-stable}] if no player wants to move to another existing or empty coalition, i.e.,
for each player~$i\in N$ and each coalition~$C\in\Gamma\cup\{\emptyset\}$ with $\vert C\vert < c$, it holds that $\Gamma(i)\succeq_i C\cup\{i\}$;
\item[\rmfamily\mdseries \textit{individually stable}] if no player can move
      to another preferred coalition without making a player in the new
      coalition worse off, i.e., for each player~$i\in N$ and for each
      coalition~$C\in\Gamma\cup\{\emptyset\}$ with $\vert C\vert < c$, it holds that
      $\Gamma(i)\succeq_i C\cup\{i\}$ or there exists a player $j\in
      C$ such that $C\succ_j C\cup\{i\}$;
\item[\rmfamily\mdseries \textit{contractually individually stable}] 
      if no
      player can move to another preferred coalition without making a player in
      the new coalition or in the old coalition worse off, i.e., for each player~$i\in N$ and for each coalition~$C\in\Gamma\cup\{\emptyset\}$ with $\vert C\vert < c$, it holds that
      $\Gamma(i)\succeq_i C\cup\{i\}$, or there exists a player $j\in C$
      such that $C\succ_j C\cup\{i\}$, or there exists a player $j'\in
      \Gamma(i)\setminus\{i\}$ such that
      $\Gamma(i)\succ_{j'}\Gamma(i)\setminus\{i\}$.
\end{description}
Note that Nash stability implies individual stability, which, in turn,
implies contractual individual stability.

\begin{figure}
	\centering
	\begin{minipage}[b]{40mm}
		\input{stability_concepts_fig.tex}
		\centering $G$
	\end{minipage}
	\qquad
	\begin{minipage}[b]{40mm}
		\input{stability_concepts2_fig.tex}
		\centering $H$
	\end{minipage}
	\caption{\reva{Two \fenhg s $G,H$ with bounded degree $4$ and $7$~players, a--g. Friend edges are solid green; enemy edges are dashed red. Assume enemies aversion, i.e., $f=1$ and $e=4$ for the utility function $u_i$ that defines the preference extension $\succeq$ (see Definition~\ref{def:fenhg} and its preliminaries). The coalition structure $\Gamma = \{ \{a, c\}$, $\{b, d, f\}, \{e, g\} \}$ is perfect in $G$ if the bound on the coalition size is $c \geq 3$, otherwise there is no perfect coalition structure. Even for unbounded coalition size, there is no perfect coalition structure in $H$. However, the coalition structure $\Gamma$ is core-stable and Nash-stable in $H$ (and due to Nash-stability, it is also individually stable and contractually individually stable).}}
	\label{fig:sc}
\end{figure}

\opt{full}{
A further popular stability concept is based on group deviation. A coalition structure $\Gamma$ is called
\begin{description}
\item[\rmfamily\mdseries \textit{core stable}] if no coalition \emph{blocks} $\Gamma$, i.e., for each non-empty coalition $C\subseteq N$ with $\vert C\vert\leq c$, there exists a player~$i\in C$ such that  $\Gamma(i)\succeq_i C$.
\end{description}
}
\reva{Figure~\ref{fig:sc} shows two examples of \fenhg s and the stability concepts defined above.}
For a stability concept, questions of interest are:
\begin{itemize}
 \item Verification: Given a game and a coalition structure, is it stable?
 \item Existence: Is a given game stable, i.e., does there exist a stable coalition structure?
\end{itemize}

\subsection{Graph Property Testing}
\label{sec:gpt}

Let $G = (\revc{N,F \cup E})$ be a graph with vertex degrees bounded by~$d$ and let $\mathcal{P}$ be a graph property, i.e., a set of graphs (e.g., let $\mathcal{P}$ be all graphs that admit a perfect coalition structure).
We say that $G$ is \emph{$\epsilon$-far} from a property~$\mathcal{P}$ if more than $\epsilon d n$ edges of $G$ have to be modified (i.e., added or removed) in order to convert it into a graph that satisfies the property $\mathcal{P}$, otherwise $G$ is \emph{$\epsilon$-close} to $\mathcal{P}$. A property tester has access to $G$ by querying a function $f_G: N \times [n] \to [n] \cup \{\mathrm{null}\}$, where $f_G(v,i)$ denotes the $i^{th}$ neighbour of $v$ if $v$ has at least~$i$ neighbours. Otherwise, $f_G(v,i) = \mathrm{null}$.

\begin{definition}[one-sided testers]
A one-sided error $\epsilon$-tester for a property $\mathcal{P}$ of bounded-degree graphs with query complexity $q$ is a randomized algorithm $\mathcal{A}$ that makes $q$ \revb{\emph{neighbourhood queries}} to $f_G$ for an input graph~$G$. The algorithm $\mathcal{A}$ accepts $G$ if $G$ has the property $\mathcal{P}$. If $G$ is $\epsilon$-far from $\mathcal{P}$, then $\mathcal{A}$ rejects $G$ with probability at least~$\nicefrac{2}{3}$.
\label{defn:one-sided}
\end{definition}

\revc{Similarly to one-sided error decision algorithms, a one-sided error tester has to accept all graphs that have the property at hand. Therefore, it has to present a witness against the property when it rejects.}

Many classic graph properties (e.g., planarity) are maintained when edges are removed from the graph; we observe the same for properties of hedonic games.

\begin{definition}[edge-monotonicity]
A graph property~$\mathcal{P}$ is \emph{edge-monotone} if for every $G=(N, F \cup E) \in \mathcal{P}$, $\{ (N, F' \cup E') \mid F' \subseteq F \wedge E' \subseteq E \} \subseteq \mathcal{P}$. In other words, every subgraph of $G$ is also in $\mathcal{P}$.
\end{definition}

\section{The Model of Property Testing for Stability Concepts}
\label{sec:pt_of_sc}

To test stability concepts, we \revc{slightly} generalize the standard edit distance \revb{that underlies the definition of being $\epsilon$-far} as follows. Since we consider graphs $G = (V, F \cup E)$ that represent \fenhg s, we have to account for the two types of edges: friends and enemies. Therefore, an edge modification is one insertion of an element to or one removal of an element from $F \cup E$, respectively, while maintaining $F \cap E = \emptyset$. In particular, turning a friend edge into an enemy edge is counted as two edge modifications (removing it from $F$ and inserting it into $E$). The intuition of these semantics is that edge modifications measure the number of compromises that are needed to reach a stable situation. If a partition is too far from being stable, too many compromises are necessary, and the partition should be discarded. Everything in-between is not an ideal situation, but may be affordable.

Now, the existence of a stable outcome in a game is modelled as a graph property as follows.

\begin{definition}[stability existence property]
  The set of stable graphs with respect to some stability concept (e.g., Nash stability) is the set of all graphs that admit a stable coalition structure.
\end{definition}

For some stability concepts, the existence of a stable outcome is guaranteed. Nevertheless, the question of whether a given partition~$\Gamma$ satisfies the stability property can still be hard to decide. For all stability concepts mentioned above, the worst case time that is needed to verify stability of $\Gamma$  is at least linear in the number of players. We can, however, tackle the following problem in sublinear time: Given a graph~$G$ and a partition of vertices~$\Gamma$, is $\Gamma$ a stable outcome for the game represented by~$G$, or is $G$ $\epsilon$-far from being a stable instance for $\Gamma$?

\begin{definition}[$\Gamma$-stability verification property]
  Let~$n \in \mathbb{N}$, and let~$\Gamma$ be a partition of~$[n]$. The set of~$\Gamma$-stable graphs with respect to some stability concept (e.g., Nash stability) is the set of $n$-vertex graphs $G$ such that~$\Gamma$ is stable for~$G$.
\end{definition}

Note that, unlike the existence of a stable coalition structure, a stability property is not closed under isomorphism as long as $\Gamma$ is not permuted accordingly.
Therefore, extending the basic model of graph property testing to reflect the semantics of hedonic games is the foundation of our main contribution, \revc{and it is the main difference between our model and the vanilla model of graph property testing}.

\begin{definition}[\revb{query access to $\Gamma$}]
	Access to $\Gamma$ is provided by a set oracle that supports two queries. A \revb{\emph{find query}} returns, given a vertex $v$, the key of the set that contains $v$. A \revb{\emph{member query}} returns, given a key~$k$ and an index $i$, the $i$-th element of the set represented by $k$, or $\mathrm{null}$ if no such element exists.
\end{definition}

\section{Property Testing Results in the FEN-Model}

In this section we study property testers for stability verification problems, resulting in Theorem~\ref{mainthm:verification}, as well as stability existence problems, resulting in Theorem~\ref{mainthm:existence}, for various individual-based stability notions within the previously defined model of \fenhg s.

\subsection{Testing Verification Problems}
\label{sec:verification}

In the following we prove Theorem~\ref{mainthm:verification},
the testability of verification problems with query complexity dependent only on the degree bound, but independent of the input graph's size, which is restated in Theorem~\ref{thm:verification-ind-stability} below.
\revb{Our algorithm can test the verification problem for every stability concept that can be defined using the notion of \emph{feasible player properties} (see Definition~\ref{def:feasible-phi}), which is what we prove in Theorem~\ref{thm:verification-tester}. We state the feasible player properties for \revc{all of the stability concepts in Table~\ref{tab:overview_concepts}} and prove the query complexity of Algorithm~\ref{alg:verification-tester} for these properties (see Theorem~\ref{thm:verification-ind-stability}).}

We observe that due to responsiveness in \revc{all preference extensions of \fenhg s}, an edge modification that benefits one player, can never be a disadvantage for other players.

\begin{observation}\label{lemma:responsive}
 Let $i,j\in N$ be two players in the \fenhg\ represented by a graph~$G=(N,F\cup E)$. Furthermore, let $\succeq_i$ be $i$'s original preference relation, and $\succeq_i'$ the preference relation of $i$ after deleting edge~$(i,j)$. Formally, the resulting game is represented by the graph $G'=(N,F \cup E \setminus \{(i,j)\})$ with the same preference extension.
 The following statements hold:
 \begin{enumerate}
  \item If $(i,j) \revb{\in} E$,
  for each $C\in\mathcal{N}_i$, $j\notin C$, it holds that $C\succ_i C\cup\{j\}$ and $C\cup\{j\}\sim_i' C$,
  \item If $(i,j) \revb{\in} F$, for each $C\in\mathcal{N}_i$, $j\notin C$, it holds that $C\cup\{j\}\succ_i C$ and $C\cup\{j\}\sim_i' C$,
  \item If $(i,j) \revb{\in} E$, and $C\succ_i D$ for two coalitions $C,D\in\mathcal{N}_i$
  with $j\in C$, it holds that $C\succ_i' D$.
  \item If $(i,j) \revb{\in} F$, and $C\succ_i D$ for two coalitions $C,D\in\mathcal{N}_i$
  with $j\notin C$, it holds that $C\succ_i' D$.
 \end{enumerate}
\end{observation}

Many stability concepts are of the form such that stability holds if and only if no player~$i$ satisfies a certain condition~$\phi(i)$. If there exists a player~$j$ that satisfies this condition~$\phi(j)$, we call $j$ a \emph{witness} for non-stability. %
Let $\phi$ denote such a player property assigning each player~$i$ either value~$1$ ($i$ is a witness against the property) or value~$0$ ($i$ is not a witness). In the following proofs we require certain conditions to hold for $\phi$ and show that \revc{all concepts in Table~\ref{tab:overview_concepts}} share these conditions, which enables us to devise a unified testing scheme.
\begin{definition}[feasible player property]\label{def:feasible-phi}
Let $\gamma$ be a stability concept and let
$\phi:\mathfrak{G}_n\times\mathfrak{C}_n \to\{0,1\}^{n}$ be a function that is parameterized by $n \in \mathbb{N}$.
If, for every $n$-player game $G \in \mathfrak{G}_n$, every coalition structure~$\Gamma\in\mathfrak{C}_n$ and every $i\in N$, it holds that $\phi_i(G,\Gamma)=0$ if and only if $\Gamma$ is stable in $G$ with respect to $\gamma$, then $\phi$ is a \emph{player property}.
We say that $\phi$ is \emph{feasible} if for every $n \in \mathbb{N}$, every $G=(N, F\cup E)\in\mathfrak{G}_n$ and every $\Gamma\in\mathfrak{C}_n$ the following conditions are met for each $i\in N$.
\begin{enumerate}[label=(\roman*)]
 \item \label{feasible:reparable} $\Gamma(i)\in\mathrm{Fav}(i) \implies \phi_i(G,\Gamma)=0$.
 \item \label{feasible:constant} The value $\phi_i(G,\Gamma)$ can be determined with a constant number of queries to the oracles for $G$ and $\Gamma$ (i.e., dependent on $\epsilon, d$, but independent of $n$).
 \item \label{feasible:responsive1} If $\phi_i(G,\Gamma)=0$ and an edge $(j,i)$, $j\notin\Gamma(i)$ is removed from $F$, resulting in a new game $G'$, it holds that $\phi_i(G',\Gamma)=0$ remains valid.
 \item \label{feasible:responsive2} If $\phi_i(G,\Gamma)=0$ and an edge $(j,i)$, $j\in\Gamma(i)$ is removed from $E$, resulting in a new game $G'$, it holds that $\phi_i(G',\Gamma)=0$ remains valid.
\end{enumerate} 
\end{definition}

\begin{table}
	\centering
	\begin{tabular}{lp{1.8cm}p{4.7cm}ll}
		\toprule
		                  & concept                          & player property: $\phi_i(G,\Gamma)=1\iff$                                                                                                                                                                                 & always ex. & is sol. for \\ \midrule
		1                 & perfection                       & $  \exists C\in\mathcal{N}_i:C\succ_i\Gamma(i)$                                                                                                                                                                           & no         & 2, 3, 4, 5, 6   \\
		\addlinespace
		2 & individual \mbox{rationality}    & $\{i\}\succ_i\Gamma(i)$                                                                                                                                                                                                   & yes        &             \\
		\addlinespace
		3 & Nash stability                   & $\exists C\in \Gamma\cup\{\emptyset\}: C\cup\{i\}\succ_i\Gamma(i)$                                                                                                                                                        & yes        & 2, 4, 5       \\
		\addlinespace
		4 & individual \mbox{stability}      & $\exists C\in \Gamma\cup\{\emptyset\}: C\cup\{i\}\succ_i\Gamma(i)$\newline and $\forall j\in C: C\cup\{i\}\succeq_j C$                                                                                                    & yes        & 2, 5           \\
		\addlinespace
		5 & contractual individual stability & $\exists C\in \Gamma\cup\{\emptyset\}: C\cup\{i\}\succ_i\Gamma(i)$\newline and $\forall j\in C: C\cup\{i\}\succeq_j C$ and\newline $\forall j'\in \Gamma(i)\setminus\{i\}: \Gamma(i)\setminus\{i\}\succeq_{j'} \Gamma(i)$ & yes        &             \\
		\addlinespace
		6 & core stability                   & $\exists C \in \mathcal{N}_i: \forall j \in C: C \succ_j \Gamma(j)$                                                                                                                                                       & no         & \revnew{2}            \\ \bottomrule
		                  &                                  &                                                                                                                                                                                                                           &            &
	\end{tabular}
	\caption{\reva{Overview of stability concepts, their player properties, whether a stable coalition structures always exists in an \fenhg, and whether a stable solution is (always) a solution for other concepts.}}
	\label{tab:overview_concepts}
\end{table}

Player properties for the stability concepts we consider are summarized in Table~\ref{tab:overview_concepts}. We note that if $\phi$ is feasible, then $\phi$ is edge-monotone. \revc{Feasibility (including edge-monotonicity) and the following two lemmas act as the tie between stability concepts we consider and property testing. In particular, they expose characteristics of the stability concepts that can be exploited by a property tester.}

In our constructions we relate to the players' favourite coalitions and make use of the following lemma that states that we can easily modify a player's local surroundings to turn the current coalition into a favourite coalition. In other words, only a constant number of compromises suffice to optimise one player's current situation. \revb{In particular, this turns a witness according to a feasible player property into a non-witness.}

\begin{lemma} \label{lemma:witness}
 Let~$G=(N,F\cup E)$ be a graph with bounded degree~$d$ and let~$\Gamma$ be a coalition structure of $N$.
 For each $i\in N$, $\mathcal{O}(d)$ queries and $d$ edge modifications are sufficient to turn $\Gamma(i)$ into one of $i$'s favourite coalitions in the \fenhg\ represented by~$G$.
\end{lemma}
\begin{proof}
 If for player~$i$, it already holds that $\Gamma(i)\in\mathrm{Fav}(i)$, no modification is required.
 Otherwise we can proceed as follows:
 Accessing the (at most~$d$) members of $N_i$ requires at most $d$ \revb{neighbourhood} queries. Moreover, we can issue one \revb{find} query each in order to determine whether a player~$j\in N_i=N_i^+\cup N_i^-$ is contained in $\Gamma(i)$.
 For each $j\in \Gamma(i)\cap N_i^-$, delete the edge $(i,j)$ from~$E$; %
 for each $j\in N_i^+\setminus\Gamma(i)$, remove the edge $(i,j)$ from~$F$.
 This requires at most $\lvert N_i^+\rvert + \lvert N_i^-\rvert\leq d$ edge modifications.
 Note that this is independent of any bound~$c$ of the coalition size.
 The obtained coalition now only contains friends of $i$ and $i$ does not have any friends outside of $\Gamma(i)$. Hence, for \revc{all preference extensions of \fenhg s}, no coalition is preferred to the current coalition $\Gamma(i)$.~\revb{\qed}
\end{proof}

\revb{The following lemma extends Lemma~\ref{lemma:witness} to multiple players. Essentially, it allows us \revnew{to} conclude that if a game is $\epsilon$-far from being stable for a stability concept that has a feasible player property, there must be roughly $\epsilon n$ witnesses.}

\begin{lemma} \label{lemma:epsilon-far}
 Let $G=(N,F\cup E) \in \mathfrak{G}_n$ be a \fenhg\ and let $\Gamma$ be a coalition structure of $N$. Furthermore, let $\gamma$ be a stability concept for which there exists a feasible player property~$\phi$.
 If there are at most $k$ witnesses, $k\cdot d$ edge modifications are sufficient to make the game stable with respect to~$\gamma$.
\end{lemma}
\begin{proof}
 By Lemma~\ref{lemma:witness}, for each witness~$i$, $d$ edge modifications are enough to turn $\Gamma(i)$ into a favourite coalition, thus, $\phi(i)=1$ is no longer satisfied.
 For each player~$j$ that is not a witness, $\phi(j)=0$ already holds which does not change due to 
 Conditions~\ref{feasible:responsive1} and~\ref{feasible:responsive2} of Definition~\ref{def:feasible-phi}.
 If there are at most $k$ witnesses, $k\cdot d$ edge modifications are sufficient such that no player satisfies $\phi(i)$, thus, stability with respect to $\gamma$ holds.~\revb{\qed}
\end{proof}

With the help of this lemma we are now ready to prove that Algorithm~\ref{alg:verification-tester} provides a property tester for the verification problem of each stability concept with a feasible player property.

\begin{algorithm}
  \caption{}
  \label{alg:verification-tester}
  \begin{algorithmic}[1]
    \Require{access to $G=(N,F\cup E)$ and $\Gamma$ is provided by an oracle, $\phi(G,\Gamma)$ is the corresponding boolean stability function}
    \Function{VerificationTester}{$N$, $F$, $E$, \revb{$\epsilon$}}
      \State $s \gets \frac{1}{\epsilon}\ln 3$\;
      \State sample $s$ players i.i.d. from $\mathcal{D}(N)$, where $\mathcal{D}(N)$ is the uniform distribution on $N$\; \label{alg:line_sampling}
      \For{each sampled player~$i$}
        \If{$\phi_i(G,\Gamma) = 1$}
          \State \Return reject
        \EndIf
      \EndFor
      \State \Return accept
    \EndFunction
  \end{algorithmic}
\end{algorithm}

\begin{theorem}\label{thm:verification-tester}
 Let $\gamma$ be a stability concept for which there exists a feasible player property~$\phi$. It holds that Algorithm~\ref{alg:verification-tester} is a one-sided error property tester for $\Gamma$-stability verification with respect to $\gamma$.
\end{theorem}
\begin{proof}
 If $\gamma$ holds, there is no witness for non-stability, i.e., for each sampled vertex $\phi(i)=0$ holds. Therefore, the tester decides in Line~6 that $\gamma$ holds with probability~$1$.
 
 If $\Gamma$ is $\epsilon$-far from being stable with respect to $\gamma$, at least $\epsilon d n$ edge modifications are required. Thus, by Lemma~\ref{lemma:epsilon-far},
 there are at least $\nicefrac{\epsilon d n}{d}=\epsilon n$ witnesses.
 Hence, the probability that a sampled player is a witness is at least
 ${\epsilon n}\cdot\nicefrac{1}{n}=\epsilon$.
 
 Then, the algorithm correctly rejects if at least one witness is sampled, i.e., the condition in Line~5 is true. The probability of this event is $1$ minus the probability that for each sampled player the condition in Line~5 is false, i.e., $\phi_i(G,\Gamma)=0$.
 The latter probability is
 at most
  \begin{align*}
    \left(1-\epsilon\right)^s
    \leq e^{-\left(\epsilon\cdot s\right)}
    =e^{-\ln 3}=\frac{1}{3}.
  \end{align*}
 Thus, the probability that the tester correctly rejects is at least $\nicefrac 2 3$.
 
 Since $\phi$ is feasible, $\phi_i(G,\Gamma)$ can be determined in constant query time. Hence, the tester requires constant query time dependent on the applied function $\phi$.~\revb{\qed}
\end{proof}

Now it remains to show that \revc{each stability concept in Table~\ref{tab:overview_concepts}} has such a feasible player property~$\phi$. By Theorem~\ref{thm:verification-tester}, Algorithm~\ref{alg:verification-tester} is a verification tester with a query complexity depending on $\phi$. \revc{As mentioned in Section~\ref{sec:gpt}, a one-sided error tester needs to present a witness when it rejects. Therefore, the query complexity of a one-sided error tester is at least the minimum size of a witness. This is, in particular, relevant for testing stability verification of core stability, as core stability puts conditions on the preferences of \emph{all} participants of a potentially blocking coalition (see the corresponding player property in Table~\ref{tab:overview_concepts}). Therefore, we parametrize by the maximum coalition size for testing stability verification of core stability.}

\begin{theorem}\label{thm:verification-ind-stability}
 For the \fenhg\ model,
 the $\Gamma$-stability verification property can be tested with respect to
 \begin{enumerate}
  \item perfection and individual rationality with query complexity in $\mathcal{O}(\nicefrac{d}{\epsilon})$,
  \item Nash stability, individual and contractual individual stability with query complexity in $\mathcal{O}(\nicefrac{d}{\epsilon})$\opt{conference}{.}\opt{full}{,}
  \opt{full}{
   \item core stability with query complexity in $\mathcal{O}(\nicefrac{c^2\cdot d^{c^2\reva{+1}}\cdot e^c}{\epsilon})$, where $c$ is the maximum coalition size.
  }
 \end{enumerate}
\end{theorem}
\begin{proof}
 For \revc{each stability concept in Table~\ref{tab:overview_concepts}} we show that there exists a feasible player property by considering the four conditions of Definition~\ref{def:feasible-phi}. This allows us to apply Theorem~\ref{thm:verification-tester}. In each case we determine the exact query complexity of the tester.
 \begin{description}
  \item[\rmfamily\mdseries \textit{perfect:}] The corresponding player property is
    \[\phi_i(G,\Gamma)=1\iff \exists C\in\mathcal{N}_i:C\succ_i\Gamma(i).\]
  \begin{description}
	  \item[\ref{feasible:reparable}] Condition~\ref{feasible:reparable} holds by definition of perfection.
	  \item[\ref{feasible:constant}] We have $\phi_i(G,\Gamma)=1$ if and only if $\Gamma(i)$ \revc{is missing one of $i$'s friends or contains one of $i$'s enemies},
	  which can be verified in constant query time
	  by asking whether $j$ is in the same coalition as $i$ for each $j\in N_i$.
	  Therefore, Condition~\ref{feasible:constant} is met with $d$ queries per sampled player. In total, the query complexity is in $\mathcal{O}(\nicefrac{d}{\epsilon})$.
	  \item[\ref{feasible:responsive1}, \ref{feasible:responsive2}]
	  Conditions~\ref{feasible:responsive1} and~\ref{feasible:responsive2} are implied immediately by Observation~\ref{lemma:responsive}, since the relation $\Gamma(i)\succeq_i C$ remains valid in each relevant case.
  \end{description}
  
  \item[\rmfamily\mdseries \textit{individually rational:}] Here, the player property is
    \[\phi_i(G,\Gamma)=1\iff \{i\}\succ_i\Gamma(i).\]
  \begin{description}
  \item[\ref{feasible:reparable}] Condition~\ref{feasible:reparable} holds, since $\revb{\Gamma(i)} \in \mathrm{Fav}(i)$ implies that $\revb{\Gamma(i)} \succeq_i\{i\}$. \item[\ref{feasible:constant}]
  Since the decision whether $\{i\}\succ_i\Gamma(i)$ only depends on $\lvert\Gamma(i)\cap N_i\rvert$, $d$ queries are sufficient, i.e., a total query complexity in $\mathcal{O}(\nicefrac{d}{\epsilon})$, which satisfies Condition~\ref{feasible:constant}.
  \item[\ref{feasible:responsive1}, \ref{feasible:responsive2}] Again, Conditions~\ref{feasible:responsive1} and~\ref{feasible:responsive2} can be implied by Observation~\ref{lemma:responsive}.
  \end{description}
  
  \item[\rmfamily\mdseries \textit{Nash-stable:}] A witness~$i$ against Nash stability satisfies
  \[
   \phi_i(G,\Gamma)=1\iff \exists C\in \Gamma\cup\{\emptyset\}: C\cup\{i\}\succ_i\Gamma(i).
  \]
  \begin{description}
	  \item[\ref{feasible:reparable}] Condition~\ref{feasible:reparable} holds, since $\Gamma(i)\in\mathrm{Fav}(i)$ implies $\Gamma(i)\succeq_i C$; \revb{in particular, this is true if $C=C'\cup\{i\}$ for a $C'\in \Gamma\cup\{\emptyset\}$}.
	  
	  \item[\ref{feasible:constant}] We consider the following cases regarding Condition~\ref{feasible:constant}:
	  
	  \begin{description}
	   \item[\mdseries(a)] If $\Gamma(i)\cap N_i^+\neq\emptyset$, $i$ wants to deviate to a coalition~$C$ with $u_i(C \cup \{ i \}) > u_i(\Gamma(i))$. Due to the linearity of the preferences, this can only be $\{ i \}$ (with $u_i(\{ i \}) = 0$) or a coalition in $\Gamma$ that contains at least one friend. There are at most $|N_i^+|\leq d$ coalitions in $\Gamma$ that contain a friend, namely $\Gamma(j)$, $j\in N_i^+$. Hence, at most $d$ comparisons of coalitions are sufficient, which can be done with at most $d$~neighbour and~$d$ find queries by Equation~\eqref{eq:comparison}.
	   \item[\mdseries(b)] If $\Gamma(i)\cap N_i^+=\emptyset$, but $N_i^+\neq\emptyset$, the analysis is analogous to (a).
	   \item[\mdseries(c)] If $N_i^+=\emptyset$ and $\Gamma(i)\cap N_i^-=\emptyset$, $\Gamma(i)$ is already one of $i$'s favourite coalitions, hence $\phi_i(G,\Gamma)=0$.
	   \item[\mdseries(d)]
	   If $N_i^+=\emptyset$ and $\Gamma(i)\cap N_i^-\neq\emptyset$, $i$ wants to deviate to the single player coalition~$\{i\}$, hence $\phi_i(G,\Gamma)=1$.  
	  \end{description}
	   
	  It can be decided with $d$ neighbour queries which of the four cases holds for $\Gamma(i)$. The at most~$d$ coalition comparisons require at most $d$ additional find queries. Therefore, $\phi_i(G,\Gamma)$ can be decided with $\mathcal{O}(d)$ queries, satisfying Condition~\ref{feasible:constant}.
	  The total query complexity of Algorithm~\ref{alg:verification-tester} is in $\mathcal{O}(\nicefrac{d}{\epsilon})$.  
	  
	  \item[\ref{feasible:responsive1}, \ref{feasible:responsive2}] Again, Conditions~\ref{feasible:responsive1} and~\ref{feasible:responsive2} are implied by Observation~\ref{lemma:responsive}.  
  \end{description}
  
  \item[\rmfamily\mdseries \textit{individually stable:}] A witness~$i$ against individual stability satisfies
  \begin{align*}
   \phi_i(G,\Gamma)=1\iff &\exists C\in \Gamma\cup\{\emptyset\}: C\cup\{i\}\succ_i\Gamma(i)
   \ \mathrm{and}\\ &\forall j\in C: C\cup\{i\}\succeq_j C.
  \end{align*}
  \begin{description}
	  \item[\ref{feasible:reparable}] Hence, if $i$ is not a witness for a Nash deviation, it cannot be a witness here, either.
	  Therefore, Condition~\ref{feasible:reparable} holds.
	  
	  \item[\ref{feasible:constant}] \revb{Since $C\cup\{i\}\succ_i\Gamma(i)$ was already discussed for the case of Nash stability, it remains to consider the additional condition $\forall j\in C: C\cup\{i\}\succeq_j C$ that extends the player property of Nash stability.}
	  If $i$ wants to deviate, this is due to one of the cases (a), (b), or (d) \revb{from the case distinction of Nash stability} above.
	  In cases (a) and (b), we have to consider at most $|N_i^+|$ candidate coalitions, $i$ can deviate to.
	  For each neutral player $j\in N_i^0$, it holds that $\Gamma(j)\sim_j\Gamma(j)\cup\{i\}$. Thus, we only have to ask $i$'s neighbours for permission to enter the new coalition, which are in total at most~$d$. In fact, due to the symmetry of preferences, friends always welcome~$i$, and enemies never do.
	  \revb{In case (d),} $\Gamma(i)$ is not acceptable, and $i$ is always welcome in $\{i\}$.
	  
	  We obtain $\phi_i(G,\Gamma)=0$ if and only if there are no enemies in $C$, which we can decide with at most $d$ queries. 
	  Thus, we can employ the same queries as for Nash stability in order to determine $\phi_i(G,\Gamma)$, which satisfies Condition~\ref{feasible:constant}. The total query complexity of Algorithm~\ref{alg:verification-tester} is in $\mathcal{O}(\nicefrac{d}{\epsilon})$.
	  
	  \item[\ref{feasible:responsive1}] If $i$ wants to move to another coalition $C\subseteq\Gamma\cup\{\emptyset\}$ but there exists a player~$j\in C$ with $C\succ_j C\cup\{i\}$, then $j$ is $i$'s enemy due to Observation~\ref{lemma:responsive}. Therefore, deleting edges from $F$ cannot make $C \cup \{ i \}$ a feasible deviation if it was not feasible before. If $i$ is not a witness because there does not exist any preferred coalition to move to, the arguments for Nash stability can be applied. Thus, Condition~\ref{feasible:responsive1} is met.
	  
	  \item[\ref{feasible:responsive2}] If an edge $(i,j)$ is deleted from $E$ and $j\in\Gamma(i)$, Condition~\ref{feasible:responsive2} can only be false if $\Gamma(j) \succ_j \Gamma(j)\cup\{i\}$. However, $i \in \Gamma(j)$ because $\Gamma(j) = \Gamma(i)$, which is a contradiction. If $i$ is not a witness because there does not exist any preferred coalition to move to, the arguments for Nash stability can be applied. Therefore, Condition~\ref{feasible:responsive2} is met.
  \end{description}
  
  \item[\rmfamily\mdseries \textit{contractually individually stable:}] \reva{A witness~$i$ against contractually individual stability satisfies
  \begin{align*}
  	\phi_i(G,\Gamma)=1\iff
  		&\exists C\in \Gamma\cup\{\emptyset\}: C\cup\{i\}\succ_i\Gamma(i)
  	   \ \mathrm{and}\\
  		&\forall j\in C: C\cup\{i\}\succeq_j C\ \mathrm{and} \\
		&\forall j'\in \Gamma(i)\setminus\{i\}: \Gamma(i)\setminus\{i\}\succeq_{j'} \Gamma(i).
  \end{align*}
  }%
  
  \begin{description}
	  \item[\ref{feasible:reparable}] Condition~\ref{feasible:reparable} holds analogously to individual stability.
	  
	  \item[\ref{feasible:constant}]
	  Observe that for neutral players~$j'\in N_i^0\cap\Gamma(i)$ it holds that $\Gamma(i)\sim_{j'} \Gamma(i)\setminus\{i\}$ and
	  for enemies~$j'\in N_i^-\revb{\cap\Gamma(i)}$ it holds that  $\Gamma(i)\setminus\{i\}\succ_{j'}\Gamma(i)$.
	  Again, if $i$ wants to deviate to a coalition, cases~(a), (b) and~(d) from Nash stability remain.
	  In case~(a) $i$ has friends in $\Gamma(i)$ that $i$ contractually depends on. Here $\phi_i(G,\Gamma)=0$.
	  In cases~(b) and~(d) there are no friends in $\Gamma(i)$, which means there is no contractual dependence. Then, $i$ is a witness against contractual individual stability if and only if it is a witness against individual stability. Thus, in case~(d) $\phi_i(G,\Gamma)=1$ and in case~(b) that same queries as above can be applied. Thus, we need at most $\mathcal{O}(d)$ queries in order to determine $\phi_i(G,\Gamma)$. The total query complexity of Algorithm~\ref{alg:verification-tester} is in $\mathcal{O}(\nicefrac{d}{\epsilon})$, which satisfies Condition~\ref{feasible:constant}.
	  
	  \item[\ref{feasible:responsive1}, \ref{feasible:responsive2}] Conditions~\ref{feasible:responsive1} and~\ref{feasible:responsive2} hold with analogous arguments as above.
  \end{description}
  
\opt{full}{
  \item[\rmfamily\mdseries \textit{core-stable:}]
  \reva{A witness~$i$ against core-stability satisfies
  \begin{equation*}
  	\phi_i(G,\Gamma)=1\iff \exists C \in \mathcal{N}_i: \forall j \in C: C \succ_j \Gamma(j).
  \end{equation*}
  }%
  Observe that if there exists a coalition blocking~$\Gamma$, its connected components also block $\Gamma$. Hence, we can assume that $C$ is connected. It either holds that $\Gamma(i)$ is not acceptable for~$i$, and $C=\{i\}$ or $C$ is a neighbouring coalition\footnote{We call $C\subseteq\mathcal{N}_i$ a \emph{neighbouring coalition} if at least one agent $j\in C$, $j\neq i$, is a neighbour of $i$ in $N_i$ (i.e., $j$ is either a friend in $N_i^+$ or an enemy in $N_i^-$).} to $i$.
  
  \begin{description}
	  \item[\ref{feasible:reparable}] Condition~\ref{feasible:reparable} holds, since $i$ cannot strictly prefer a blocking coalition to a favourite coalition.
	  
	  \item[\ref{feasible:constant}]
	  In order to detect a connected blocking coalition~$C$ that contains~$i$, we have to query the oracle as follows:
	  We verify whether $\Gamma(i)$ is acceptable for~$i$ by considering the intersection of $i$'s neighbours $N_i$ with $\Gamma(i)$. If it is not acceptable, $C=\{i\}$ holds, and thus, $\phi_i(G,\Gamma)=1$.
	  If $\Gamma(i)$ is acceptable, $C$ can only be any of $i$'s neighbouring coalitions that are connected and of size $\lvert C\rvert \leq c$.
	  There are at most
	  \reva{
	  \[
	   \sum_{\reva{\ell}=1}^c\binom{d^c}{\reva{\ell}}\leq\sum_{\reva{\ell}=1}^c\left(\frac{d^c\cdot e}{\reva{\ell}}\right)^\reva{\ell}\leq c\cdot d^{c^2}\cdot e^c
	  \]
	  }%
	  possible such coalitions.
	  For each of these coalitions~$C$, we ask the at most $c$ contained players~$j\in C$ whether they prefer $C$ to $\Gamma(j)$ \revnew{by querying their at most $d$ neighbours}. This is sufficient to verify whether one of the coalitions blocks $\Gamma$ and, thus, whether $i$ is a witness. \revb{Therefore, the total query complexity of Algorithm~\ref{alg:verification-tester} is in
	  $\mathcal{O}(\nicefrac{c^2\cdot d^{c^2\reva{+1}}\cdot e^c}{\epsilon})$} and Condition~\ref{feasible:constant} holds.
	  
	  \item[\ref{feasible:responsive1}, \ref{feasible:responsive2}]
	  Let $\phi_i(G,\Gamma)=0$, i.e., for each neighbouring coalition~$C$ to $i$, there exists some $j\in C$ such that $\Gamma(j)\succeq_j C$.
	  Condition~\ref{feasible:responsive1} holds by the following argument.
	  Assume, an edge $(i,k)$ is removed from $F$, $k\in\Gamma(i)$. Then, it holds that:
	  If $k\neq j$ and $i\neq j$, $\phi_i(G,\Gamma)=0$ remains valid.
	  If $i=j$, $\Gamma(i)\succeq_i C$ implies that $\Gamma(i)\succeq'_i C$ by Observation~\ref{lemma:responsive}.
	  If $k=j$ and $i\neq j$, $\Gamma(k)\succeq_k C$ implies that $\Gamma(k)\succeq'_k C$ by Observation~\ref{lemma:responsive}.
	  Similarly, Condition~\ref{feasible:responsive2} is obtained by Observation~\ref{lemma:responsive}.
	  \opt{jaamas}{\revb{\qed}}
  \end{description}
}
 \opt{aamas}{\qedhere}
 \opt{arxiv}{\qedhere}
 \end{description}

\end{proof}

\subsection{Testing Existence Problems}
\label{sec:existence}
Now we prove Theorem~\ref{mainthm:existence}.
In general, there always exists an individually rational
coalition structure \revb{(all coalitions are singletons)}.
\citet{bog-jac:j:stability-hedonic-coalition-structures} show that in symmetric additively separable hedonic games there always exists a Nash-stable coalition structure. Note that in our model the argument that if a player deviates, the social welfare increases, remains valid, even for bounded coalition size. \reva{Their result and its proof extend easily to the case of bounded coalitions sizes, and a sketch of the proof is provided for the sake of referencing to parts of it later.}
\begin{lemma}[\citep{bog-jac:j:stability-hedonic-coalition-structures}]
	\label{thm:always_nash_stable}
	Each symmetric \fenhg\ (\revc{such as all preference extensions of \fenhg s}) allows a Nash-stable, and consequently individually stable and contractually individually stable coalition structure.
\end{lemma}
\begin{proof}[sketch]
	Let $\Gamma$ be a coalition structure containing coalitions $\Gamma(i)$ and $C\subseteq N$ with $\lvert C\rvert\leq c-1$. We assume that $C\cup \{i\}\succ_i\Gamma(i)$.
	Moreover, let $\Gamma'$ be the coalition structure obtained if $i$ deviates to $C$, i.e., $\Gamma'$ contains $\Gamma(i)\setminus\{i\}$ and $C\cup\{i\}$.
	The social welfare $\mathrm{SW}(\Gamma)$ of a coalition structure~$\Gamma$ is the sum of all players' utilities of their current coalition. We observe that the difference of the social welfare of $\Gamma'$ and $\Gamma$ always increases, which means that there exists a local maximum resulting in a Nash-stable coalition structure.
	It holds that the difference $\mathrm{SW}(\Gamma')-\mathrm{SW}(\Gamma)$ equals
	\begin{align}
	&\sum_{j\in N}\left(u_j(\Gamma'(j))-u_j(\Gamma(j))\right) \notag \\
	=\,&\underbrace{u_i(C\cup\{i\})-u_i(\Gamma(i))}_{>0}
	+\sum_{j\in C\cap N_i}\left(u_j(C\cup\{i\})-u_j(C)\right) \notag
	\\&
	+\sum_{j\in \Gamma(i)\cap N_i}\left(u_j(\Gamma(i)\setminus\{i\})-u_j(\Gamma(i))\right) %
	\notag \\
	=\,&\underbrace{u_i(C\cup\{i\})-u_i(\Gamma(i))}_{>0}
	+f\cdot\lvert C\cap N_i^+\rvert-e\cdot\lvert C\cap N_i^+\rvert \label{eq:bogo-needs-symmetry}
	\\
	&-f\cdot\lvert \Gamma(i)\cap N_i^+\rvert+e\cdot\lvert \Gamma(i)\cap N_i^+\rvert
	> 0 \notag
	\end{align}
	Hence, there always exists a Nash-stable coalition structure, even if the coalition size is bounded. Since Nash stability implies individual and contractual individual stability, they are guaranteed to exist as well. \revb{\qed}
\end{proof}

There does not necessarily exist a perfect coalition structure. For example, there does not exist any perfect coalition structure for $G=(\{1,2,3\},F \cup E)$, where $F=((1,2), (2,3))$, $E=(3,1)$. On the other hand, if $F = \emptyset$, the utility of any coalition structure is at most $0$, so singleton coalitions are perfect; if $E = \emptyset$, there exists a perfect coalition structure if and only if no connected component is larger than $c$. For the general case $\lvert E \rvert, \lvert F \rvert \geq 0$, we show that there exists a tester with one-sided error.

\begin{theorem}\label{thm:perfection-bounded}
 There is a one-sided error property tester
 with constant query complexity
 for the existence of a perfect coalition structure
 in the \fenhg\ model with a constant coalition size bound $c$.
\end{theorem}
\begin{proof}
 Let $v \in N$, and observe that $C$ is a favourite coalition of $v$ if and only if $C \cap N^+_v = N^+_v$ and $C \cap N^-_v = \emptyset$. It follows that there exists a perfect coalition structure $\Gamma$ if and only if there does not exist any edge in $E$ between vertices of the same connected component of $G[N_F]$, where $N_F$ is the set of endpoints in $F$, i.e., $N_F = \{u \mid (u,v) \in F \}$. This suggests the following algorithm: first, sample a set $S$ of $\lvert S \rvert = \nicefrac{1}{\epsilon} \ln 3$ vertices at random. \revb{For each $v \in S$, we run a BFS that follows only edges in $F$. If, at any point during its execution, the BFS has explored more than $c$ vertices or has visited two endpoints \reva{$u,w$ of the same edge $(u,w) \in E$}, the tester rejects.\footnote{\revc{We note that the algorithm does not wrongly detect infeasible blocking coalitions when rejecting after exploring $c$ vertices. A necessary condition for a coalition structure to be perfect is that for every connected component $H$ in $(N,F)$, all vertices of $H$ must be in the same coalition. If $H$ is larger than $c$, it follows that there is no coalition structure with bounded coalition size $c$ such that every player is in one of their favourite coalitions. In other words, $H$ is a witness.}} Otherwise, i.e., if none of these two conditions is met during the BFS for any $v \in S$, the tester accepts the graph.}
 
 By the above observation, every path in $G$ that contains only edges from $F$ must be in the same coalition in a perfect coalition structure. The algorithm rejects only when it finds a path $P$ such that for every coalition structure $\Gamma$ such that some coalition $C \in \Gamma$ contains $P$, $C$ also contains $(u,v) \in E$, which is a witness against the existence of a perfect coalition structure.
 
 If $G$ is $\epsilon$-far from having a perfect coalition structure, then at least $\epsilon d n$ edges in $F \cup E$ have to be \revc{\emph{removed}} in order to make $G$ have a perfect coalition structure because having a perfect coalition structure is an edge-monotone property \revb{(i.e., adding edges never decreases the graph's distance to the property, see also Lemma~\ref{lemma:epsilon-far})}. Let $R$ be a minimal set of edges that have to be removed. Since every vertex is incident to at most $d$ other vertices, at least $2 \lvert R \rvert / d > 2 \epsilon d n / d > \epsilon n$ vertices must be incident to an edge from $\lvert R \rvert$. If a vertex that is incident to an edge in $R$ is in $S$, the algorithm finds a witness. The probability that none of the vertices in $S$ is incident to an edge in $R$ is at most
 \begin{equation}
 	\left(1-\frac{\epsilon n}{n}\right)^{\frac{1}{\epsilon} \ln 3} \leq \frac{1}{3}. \label{eq:perfection_bound_ineq}
 \end{equation}
 As argued above, if a vertex in $S$ is incident to an edge from $R$, the tester finds $(u,v)$ and rejects. \revb{\qed}
\end{proof}

\subsection{Extensions to Weighted and Directed Encodings}
\label{sec:extensions}

\paragraph{Weighted, Undirected \revc{FEN-Encodings}.}
The resulting games are equivalent to \emph{additively separable} hedonic games~\citep{bog-jac:j:stability-hedonic-coalition-structures}. If each edge contributes equally to the edit distance \revb{(that underlies the definition of $\epsilon$-far)}, this does not affect our proofs and Lemma~\ref{thm:always_nash_stable} because they rely on the linearity of the utility function only. If an edge contributes proportional to its weight, we can use the following standard techniques from property testing. \revb{The first option is} to require that the weights are bounded by some value $W$ so that we can simply increase the sampling size in our algorithms by a factor of $W$. To see why this works, imagine an edge with weight $w$ as $w$ parallel edges, which essentially increases the bound on the vertex degrees to $W \cdot d$. On the other hand, unbounded weights cannot be handled in the standard model by constant-query testers because a single edge that has weight $2 \epsilon d n$ can make a graph $\epsilon$-far, yet it is very unlikely to find this edge by sampling $O(1)$ vertices uniformly. Therefore, another option is to allow vertex sampling proportional to the weights of incident edges or, equivalently\footnote{\revb{The following way to convert from one into the other is known:} (i) Sample a vertex proportional to its weight and pick an incident edge with probability proportional to the edge's weight relative to the weight of all incident edges, (ii) sample an edge proportional to its weight, and pick an incident vertex uniformly at random.}, sampling edges proportionally to their weight.

\opt{full}{
	We describe the modifications that are required to make our algorithms work with edge-weighted graphs and weight-proportional vertex sampling. Let $G = (N,F \cup E)$ be a graph with edge weights $w : F \rightarrow \ (0, \infty)$ and $w : E \rightarrow \ (-\infty, 0)$, and define the weight of a vertex $v$ as $w(v) = \sum_{e \in N_v} \lvert w(e) \rvert$. For the sake of simplicity, we write $w : F \cup E \rightarrow (0, \infty) \cup (-\infty, 0)$ in the following.
	
	\revc{
	\begin{definition}[weight-aware graph distance]
		\label{defn:weight-dist}
		A graph $G = (N, F \cup E)$ with weights $w : F \cup E \rightarrow (0, \infty) \cup (-\infty, 0)$ is $\epsilon$-far from an edge-monotone\footnote{All properties that we consider are edge-monotone. Note that it is not clear how to define $\epsilon$-far in general for properties that are not edge-monotone because one would need to assign weights to added edges. To minimize the distance of $G$ to $\mathcal{P}$, one could always choose the smallest weight, i.e., $1$, but this might break the semantics of the weights.} property $\mathcal{P}$ if there is no set of edges $R \subseteq F \cup E$ with weight $w(R) := \sum_{e \in R} \lvert w(e) \rvert \leq \epsilon \cdot \lvert w(F \cup E) \rvert$ such that the graph $G' = (V, (F \cup E) \setminus R)$ is in $\mathcal{P}$.
	\end{definition}
	
	\begin{definition}[weight-proportional distribution]
		\label{defn:weight-proportional}
		Given a set $S$ that is weighted by a function $w : S \rightarrow (0, \infty)$, the weight-proportional distribution $\mathcal{W}_w(S)$ assigns probability mass $w(x) / w(S)$ to $x \in S$, i.e., if $X \sim \mathcal{W}_w(S)$ is a random variable, then $\Pr[X = s] = w(x) / \sum_{y \in S} w(y)$.
	\end{definition}
	
	\begin{definition}[weight-aware tester]
		\label{def:weight-tester}
		A wait-aware one-sided error tester for graphs $G = (N, F \cup E)$ with edge weights $w : F \cup E \rightarrow (0, \infty) \cup (-\infty, 0)$ is a one-sided error tester (see Definition~\ref{defn:one-sided}) that has oracle access to i.i.d. samples from the distribution $\mathcal{W}_w(N)$.
	\end{definition}
	}
	
	\paragraph{Weighted Verification Problems.}
	Lemma~\ref{lemma:witness} and Lemma~\ref{lemma:epsilon-far} hold without modifications for edge-weighted graphs. We modify Algorithm~\ref{alg:verification-tester} as follows: in Line~\ref{alg:line_sampling}, we sample $s$ players from $N$ proportional to the sum of weights of their incident edges, i.e., the sum of weights of the relations to their friends and enemies. \revc{The resulting algorithm is given by Algorithm~\ref{alg:weighted-verification-tester}.}
	
	\begin{algorithm}
	  \caption{}
	  \label{alg:weighted-verification-tester}
	  \revc{
	  \begin{algorithmic}[1]
	    \Require{access to $G=(N,F\cup E)$, $w : F \cup E \rightarrow (0, \infty) \cup (-\infty, 0)$ and $\Gamma$ is provided by an oracle, $\phi(G,\Gamma)$ is the corresponding boolean stability function}
	    \Function{VerificationWeightAwareTester}{$N$, $F$, $E$, $w$, \revb{$\epsilon$}}
	      \State $s \gets \frac{1}{\epsilon}\ln 3$
	      \State sample $s$ players i.i.d. from $\mathcal{W}_w(N)$
	      \For{each sampled player~$i$}
	        \If{$\phi_i(G,\Gamma) = 1$}
	          \State \Return reject
	        \EndIf
	      \EndFor
	      \State \Return accept
	    \EndFunction
	  \end{algorithmic}
	  }
	\end{algorithm}
	
	\begin{theorem}[Theorem~\ref{thm:verification-tester}, weighted version]
		\label{thm:weighted-verification-tester}
		Let $\gamma$ be a stability concept for which there exists a feasible player property~$\phi$. It holds that \revc{Algorithm~\ref{alg:weighted-verification-tester}} is a \revc{weight-aware} one-sided error property tester for $\Gamma$-stability verification with respect to $\gamma$.
	\end{theorem}
	\begin{proof}
		The modified algorithm still accepts all stable coalitions $\Gamma$. If~$\Gamma$ is $\epsilon$-far from being stable with respect to $\gamma$, every edge set $R \subseteq F \cup E$ whose removal stabilizes $\Gamma$ has weight at least $\epsilon \sum_{e \in F \cup E} w(F \cup E)$ by the weight-aware definition of $\epsilon$-far from above. Hence, the probability that a sampled player is a witness is at least
		\begin{equation*}
			\frac{\sum_{v \in \{x, y \mid (x,y) \in R \}} \lvert w(v) \rvert}{w(V)}
			\geq \frac{2 \epsilon \sum_{e \in F \cup E} w(F \cup E)}{2 \sum_{e \in F \cup E} w(F \cup E)}
			= \epsilon.
		\end{equation*}
		The proof continues as the proof of Theorem~\ref{thm:verification-tester}. \revb{\qed}
	\end{proof}
	
	\paragraph{Weighted Existence Problems.}
	Lemma~\ref{thm:always_nash_stable} still holds for symmetric, additively separable hedonic games. The tester for the existence of a perfect coalition structure can be generalized as follows.
	
	\begin{theorem}[Theorem~\ref{thm:perfection-bounded}, weighted version]
		\label{thm:weighted-perfection-bounded}
		There is a \revc{weight-aware} one-sided error tester with constant query complexity for the existence of a perfect coalition structure in the \fenhg\ model with a constant coalition size bound $c$.
	\end{theorem}
	\begin{proof}
		Let $G = (N, F \cup E)$ be the input graph. Consider the following algorithm: first, sample a set $S$ of $\lvert S \rvert = \nicefrac{1}{\epsilon} \ln 3$ vertices from $\mathcal{W}_w(N)$. \revb{For each $v \in S$, we run a BFS that follows only edges in $F$. If, at any point during its execution, the BFS has explored more than $c$ vertices or has visited two endpoints \reva{$u,w$ of the same edge $(u,w) \in E$}, the tester rejects. Otherwise, i.e., if none of these two conditions is met during the BFS for any $v \in S$, the tester accepts the graph.}
		
		As in the proof of Theorem~\ref{thm:perfection-bounded}, we can argue that the algorithm never rejects a graph that has a perfect coalition. Therefore, let $G$ be $\epsilon$-far from having a perfect coalition, let $R \subseteq F \cup E$ be such that $(N, (F \cup E) \setminus R)$ is in $\mathcal{P}$ and $w(R)$ is minimal among all such sets. We only need to prove that a properly adapted version of Inequality~\ref{eq:perfection_bound_ineq} holds. Observe that $\sum_{v \in \{x, y \mid (x,y) \in R \}} \lvert w(v) \rvert \geq 2 \epsilon \cdot w(F \cup E)$ by the weight-aware definition of $\epsilon$-far \revc{(see Definition~\ref{defn:weight-dist})}. Therefore, the probability that none of the vertices in $S$ is incident to an edge in $R$ is at most
		\begin{equation*}
			\left(1-\frac{\sum_{v \in \{x, y \mid (x,y) \in R \}} \lvert w(v) \rvert}{w(N)}\right)^{\frac{1}{\epsilon} \ln 3}
			\leq \left(1-\frac{2 \epsilon \cdot w(F \cup E)}{2 \cdot w(F \cup E)}\right)^{\frac{1}{\epsilon} \ln 3}
			\leq \frac{1}{3}. \tag*{\revb{\qed}}
		\end{equation*}
	\end{proof}
}

\paragraph{Directed, Unweighted \revc{FEN-Encodings}.}
\opt{full}{In this case, friendship relations are no longer necessarily symmetric. For instance, this can be the case in social networks with a follower system.}
In property testing, there are two different models of directed graphs.

In the \revb{\emph{bidirectional model}}, we may see all incoming and outgoing edges\opt{full}{ (followed players and followers can be queried)}. Assuming bounded out-degree, verification is not affected by the bidirectional model.

In the \revb{\emph{unidirectional model}}, we may see only one type - usually outgoing edges\opt{full}{ (only followed players can be queried)}. For the unidirectional model with outgoing edges,
\opt{full}{
the verification problems remain testable.
While being a witness against perfection, individual rationality, or Nash stability only depends on players' own preferences,}
the query complexity for 
individual stability and contractually individual stability depends on (the minimum of) the in-degree or the maximum coalition size: the case analysis changes because we are required to evaluate the preferences of the other members of the (still at most $O(d)$) affected coalitions.

\opt{full}{%
Although there is a lossy but still sublinear transformation from constant-query testers in the bidirectional model into the unidirectional model~\cite{CzuRel16}, the unidirectional model is often much harder to analyse and yields testers with worse query complexity. In detail, we obtain the following results.

\paragraph{Directed Verification Problems.}
Observation~\ref{lemma:responsive}, Definition~\ref{def:feasible-phi}, and Lemma~\ref{lemma:witness} are stated from a single player's perspective. Hence they remain valid \revc{assuming the standard notation that $(i,j)$ represents an arc from $i$ to $j$}. Lemma~\ref{lemma:epsilon-far}, Algorithm~\ref{alg:verification-tester}, and Theorem~\ref{thm:verification-tester} are unaffected by the generalisation. In fact the changes only depend on the player properties of the stability notions.
\begin{theorem}[Theorem~\ref{thm:verification-ind-stability}, directed, unidirectional version]
 \label{thm:verification-directed}
 For the \fenhg\ model,
 the $\Gamma$-stability verification property can be tested with respect to
 \begin{enumerate}
  \item perfection, individual rationality, and Nash stability with query complexity in $\mathcal{O}(\nicefrac{d}{\epsilon})$,
  \item individual and contractual individual stability with query complexity in $\mathcal{O}(\nicefrac{d(d+c)}{\epsilon})$,
  \item core stability with query complexity in $\mathcal{O}(\nicefrac{c^2\cdot d^{c^2\revnew{+1}}\cdot e^c}{\epsilon})$, where $c$ is the maximum coalition size.
 
 \end{enumerate}
\end{theorem}
\begin{proof}
  Conditions~\ref{feasible:reparable}, \ref{feasible:responsive1}, and~\ref{feasible:responsive2}
  remain valid for the unidirectional case.
  For perfection, individual rationality, and Nash stability Condition~\ref{feasible:constant} still holds with the same query complexity, since only queries for outgoing edges are made for each potential witness.
  
  For individual stability, there are still at most $\mathcal{O}(d)$ possible coalitions, a sampled player~$i$ might want to deviate to. Whenever player~$i$ wants to deviate to a new coalition~$C$, now each player~$j$ in~$C$ needs to be asked whether $(j,i)$ is an edge in $E$. If, for all $j\in C$, that is not the case, $i$ is a witness. This requires $\mathcal{O}(d(d+c))$ queries, which satisfies Condition~\ref{feasible:constant}. The total complexity of Algorithm~\ref{alg:verification-tester} is in $\mathcal{O}(\nicefrac{d(d+c)}{\epsilon})$.
  
  Similarly, for contractually individual stability, a player~$k$ only lets $i$ move to a new coalition, if $(k,i)$ is not an edge in $F$. If for each $k\in\Gamma(i)$ that is the case, $i$ is a witness.
  
  For core stability, we already employ all necessary queries to new coalition members in the proof of the undirected case. Therefore, the query complexity remains the same.\revb{\qed}
\end{proof}

\paragraph{Directed Existence Problems.}
}

Lemma~\ref{thm:always_nash_stable} is not true anymore in general (\revb{note that Equation~\eqref{eq:bogo-needs-symmetry} does not hold in general due to asymmetry}; see~\cite{bog-jac:j:stability-hedonic-coalition-structures} for an example). 
A perfect coalition structure exists if the weakly connected components induced by $F$ do not induce any edges from $E$. Exploring weakly connected components causes no problem in the bidirectional model, but in the unidirectional model it can render exploration almost impossible. For example, consider a cycle of length $c$ with alternating edge direction and all but one edge being friend edges. This graph makes it impossible to form perfect coalitions, and having $\epsilon d n$ copies of it as subgraphs makes a graph $\epsilon$-far from admitting a perfect coalition structure.
\opt{conference}{%
	For example, consider a cycle of length $c$ with alternating edge direction and all but one edge being friend edges. This graph makes it impossible to form perfect coalitions, and having $\epsilon d n$ copies of it as subgraphs makes a graph $\epsilon$-far from admitting a perfect coalition structure.
}%
\opt{full}{%
In fact, we show that testing the existence of a perfect coalition structure requires almost linear query complexity.
\begin{theorem}
	\label{thm:perfection-bound-unidirectional}
	Let $\epsilon > 0$ and $c \geq \nicefrac{1}{\epsilon d}$. There is no one-sided error $\epsilon$-tester for the existence of a perfect coalition structure with coalition size bound~$c$ in the unidirectional model that uses \reva{$o(n^{1 - 2 \epsilon d})$} queries.
\end{theorem}
\begin{proof}
	\revm{
	For all even $\ell \geq 6$, we define an \emph{alternating-cycle gadget} $H_{\ell}=(V_{\ell},F_{\ell} \cup E_{\ell})$, where $V_{\ell} = \{ 0, \ldots, \ell-1 \}$, $F_{\ell} = \{ (i, i-1) \mid i \in \{1, 3, \ldots, \ell-1 \} \} \cup \{ (i, i+1) \mid i \in \{1, 3, \ldots, \ell-3 \} \}, E_{\ell} = \{(\ell-1, 0)\}$ (see Fig.~\ref{fig:alternating_cycle}). Let $n \geq \ell$; we construct the graph $G_{\ell}$
	\revnew{as the union of $\lfloor \nicefrac{n}{\ell} \rfloor$ copies of $H_{\ell}$} and $n - \ell \cdot \lfloor \nicefrac{n}{\ell} \rfloor$ isolated vertices. Without loss of generality, to \revnew{simplify} the argument from now on, we assume that $n$ is divisible by $\ell$ and that $\ell$ is even. The graph $G_{\ell}$ is $\epsilon$-far from having a perfect coalition structure if $\nicefrac{n}{\ell} > \epsilon d n \Leftrightarrow \ell < \nicefrac{1}{\epsilon d}$.
	
	\begin{figure}
		\centering
\begingroup%
  \makeatletter%
  \providecommand\color[2][]{%
    \errmessage{(Inkscape) Color is used for the text in Inkscape, but the package 'color.sty' is not loaded}%
    \renewcommand\color[2][]{}%
  }%
  \providecommand\transparent[1]{%
    \errmessage{(Inkscape) Transparency is used (non-zero) for the text in Inkscape, but the package 'transparent.sty' is not loaded}%
    \renewcommand\transparent[1]{}%
  }%
  \providecommand\rotatebox[2]{#2}%
  \newcommand*\fsize{\dimexpr\f@size pt\relax}%
  \newcommand*\lineheight[1]{\fontsize{\fsize}{#1\fsize}\selectfont}%
  \ifx\svgwidth\undefined%
    \setlength{\unitlength}{96.84282798bp}%
    \ifx\svgscale\undefined%
      \relax%
    \else%
      \setlength{\unitlength}{\unitlength * \real{\svgscale}}%
    \fi%
  \else%
    \setlength{\unitlength}{\svgwidth}%
  \fi%
  \global\let\svgwidth\undefined%
  \global\let\svgscale\undefined%
  \makeatother%
  \begin{picture}(1,0.88173807)%
    \lineheight{1}%
    \setlength\tabcolsep{0pt}%
    \put(0,0){\includegraphics[width=\unitlength,page=1]{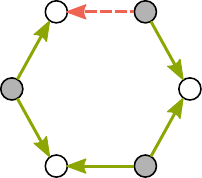}}%
    \put(0.69847207,0.79650019){\makebox(0,0)[lt]{\lineheight{1.25}\smash{\begin{tabular}[t]{l}$5$\end{tabular}}}}%
    \put(0.91762646,0.41340491){\makebox(0,0)[lt]{\lineheight{1.25}\smash{\begin{tabular}[t]{l}$4$\end{tabular}}}}%
    \put(0.69769939,0.03063077){\makebox(0,0)[lt]{\lineheight{1.25}\smash{\begin{tabular}[t]{l}$3$\end{tabular}}}}%
    \put(0.25397796,0.03120425){\makebox(0,0)[lt]{\lineheight{1.25}\smash{\begin{tabular}[t]{l}$2$\end{tabular}}}}%
    \put(0.03603763,0.41383681){\makebox(0,0)[lt]{\lineheight{1.25}\smash{\begin{tabular}[t]{l}$1$\end{tabular}}}}%
    \put(0.25731621,0.79705379){\makebox(0,0)[lt]{\lineheight{1.25}\smash{\begin{tabular}[t]{l}$0$\end{tabular}}}}%
  \end{picture}%
\endgroup%

		\caption{\label{fig:alternating_cycle} \revb{The alternating-cycle gadget $H_6$. Observe that the \revnew{odd} (grey) vertices can only be found by directly hitting them as there are no incoming edges at \revnew{odd} vertices that can be followed.}}
	\end{figure}
	
	A one-sided error tester $\mathcal{T}$ needs to sample all odd vertices of an alternating-cycle gadget to certify that $G_{\ell}$ does not have a perfect coalition structure. Without loss of generality, we may assume that whenever $\mathcal{T}$ queries a vertex~$v$, it gets \revnew{all vertices reachable from} $v$ in return (not just the $i^{th}$ neighbour of $v$ it queried for).
	This modification makes the queries that~$\mathcal{T}$ can issue only stronger but since \revnew{at most $2$ vertices are reachable from $v$}, this reduces the query complexity of the tester by at most a factor $\nicefrac{1}{3}$. Therefore, whenever $\mathcal{T}$ queries for a vertex $v$, the answer is either an \revnew{isolated vertex $\{ v \}$} if $v$ is an even vertex, or \revnew{a triple $\{u, v, w\}$} if $v$ is an odd vertex. Note that due to the strengthening of queries, querying $u$ or $w$ after querying $v$ does not reveal more information to $\mathcal{T}$. \revnew{We use Yao's principle and prove the theorem by lower-bounding the query complexity of any deterministic algorithm on the following input distribution. We generate all feasible vertex labelings of $G_{\ell}$ (isomorphic copies of $G_\ell$) and construct the uniform distribution $\mathcal{D}$ over these graphs. To simplify the analysis, we may assume that the labels are \emph{lazily} assigned to an initially unlabeled version of $G_\ell$, i.e., when the vertex corresponding to a label is queried by or returned to $\mathcal{T}$ for the first time. It is known that the resulting random graph is drawn from $\mathcal{D}$ \cite[Proof of Claim 5.5.2]{goldreich2011proximity}.}
	
	Consider running a tester $\mathcal{T}$ for perfection on a graph $G$ sampled from~$\mathcal{D}$. \revnew{Let $s$ be an upper bound on the query complexity of $\mathcal{T}$, and without loss of generality, assume that $s \leq \nicefrac{n}{6}$.} Let $S = (t_1, \ldots, t_s)$ be the sequence of all vertices from $V$ that are queried by~$\mathcal{T}$. Fix an arbitrary alternating-cycle gadget $H$ in $G$. \revnew{Intuitively, to discover all odd vertices of $H$, an odd vertex must be revealed at $\ell / 2$ positions of $S$. There are $\binom{s}{\ell/2}$ ways to choose these positions. The probability to discover an odd vertex is roughly $\nicefrac{\ell/2}{n}$, but since the set of unknown vertices is reduced by every query and the number of unknown \emph{odd} vertices from $H$ decreases, we take slightly more care.}
	
	Formally, for $j \geq 1$, consider the $j^{th}$-smallest $i$ such that $t_{i} \in V(H)$ and $t_{i}$ is an odd vertex. When this query is issued, the probability that it reveals \revnew{an odd} vertex from $H$ is \revnew{$(\ell/2 - (j-1)) / (n-k_{i})$}, where $k_i$ is the number of assigned labels before this query.  By repetitively applying \revnew{the definition of conditional probability to the event that the $j^{th}$ odd vertex of $H$ is discovered conditioned on the $1^{st}$ to the $(j-1)^{th}$ odd vertex being discovered, we can rearrange the probability that all odd vertices of $H$ are revealed to obtain a product of the former probabilities}. It follows that the probability that $S$ contains all $\ell/2$ odd vertices of $H$ is at most
	\begin{equation*}
		\revnew{
		\binom{s}{\ell/2} \cdot \prod_{j=1}^{\ell/2} \frac{\ell/2 - (j-1)}{n - 3s}
		\leq \left( \frac{se}{\ell/2} \cdot \frac{\ell/2}{n - 3s} \right)^{\ell/2}
		\stackrel{(\star)}{\leq} \left( \frac{6s}{n} \right)^{\ell/2-1}.
		}
	\end{equation*}
	\revnew{Since, for all $i \in [s]$, we have $k_i \leq 3s$, the subtrahend $3s$ accounts for the fact that each query may reveal at most $3$ vertices to $\mathcal{T}$. Inequality~$(\star)$ follows from the assumption $s \leq \nicefrac{n}{6}$.}
	
	For any $i \in [s]$, let $E_i$ be the event that \revnew{$S$} contains all odd vertices of the alternating-cycle gadget that $t_i$ is part of. By the union bound over $(E_i)_{i \in [s]}$, $S$ contains all odd vertices of at least one alternating-cycle gadget with probability at most $s \cdot ( \frac{6s}{n} )^{\ell/2-1}$. We conclude that
	\begin{align*}
		&& s \cdot \left( \frac{6s}{n} \right)^{\ell/2-1}
			&< \frac{1}{3} \\
		&\Leftrightarrow& s
			&< \left( \frac{1}{3} \left( \frac{n}{6} \right)^{\ell/2-1} \right)^{\frac{1}{\ell/2}} \\
		&\Leftarrow& s
			&< \left( \frac{n}{18} \right)^{\frac{\ell/2-1}{\ell/2}} \\
		&\Leftarrow& s
			&< \frac{1}{18} \cdot n^{1- 2 \epsilon d} .
	\end{align*}
	Assuming $s \in o(n^{1-2 \epsilon d})$ contradicts the assumption that $\mathcal{T}$ is a tester with error probability at most $\nicefrac{1}{3}$. \revb{\qed}
	}
\end{proof}
}%

\section{Open Questions}

A natural question that is related to \emph{finding} stable partitions is the following: Given a graph $G$ and a partition~$\Gamma$, is \emph{the partition~$\Gamma$} far from being stable in $G$ (instead of the graph being far from $\Gamma$-stable)? This can be generalized further: Property testing is a special case of \emph{local computation algorithms} (LCA), where one shall provide oracle access to a solution, given oracle access to the input. In property testing, the solution is a single bit (accept or reject). While it is beyond the scope of sublinear algorithms to actually compute a stable partition, one may seek to develop an LCA that gives oracle access to it.

Generalizing the results we obtained, one may seek to obtain sublinear algorithms for games with unbounded coalition size. Here, the main difficulty is to obtain insights into the local structure of very large, say, linear sized coalitions. Following a slightly different line of thought, one may consider other graphs models like the dense model, where (almost) all players relate to each other and one may ask how two players $i,j$ relate, or the general graph model, where vertices have arbitrary degrees. These models are quite different from the bounded-degree model, as well as from each other. For the dense graph model, a characterization using Szemerédi regular partitions is known \cite{AloCom09}, and it seems possible that coalition formation could be expressed in terms of regular partitions. Much less is known about the general model. So far, most research is still focused on elementary graph problems like counting the number of constant-size cliques \cite{EdeApp18}. It would be interesting to see whether stability properties are also constant-query testable in these two models. Another direction for further studies is property testers with two-sided error. These testers do not need to provide a witness against the property, but rather \emph{sufficient statistics} that a graph is far from a property with constant probability.

As mentioned in the introduction, there exist also plenty of other stability concepts like
\opt{conference}{%
  (strict) core stability,
}%
Pareto-optimality and popularity
that can operate on the same preference extension, which may be interesting to analyse in order to obtain a deeper understanding of locality mechanics in \fenhg s.
Here, the main difficulty is to circumvent the usually high computational complexity of the exact decision problems.

Finally, one may study other models of hedonic games, in particular with ordinal preferences (e.g., rankings over known edges~\cite{lan-rey-rot-sch-sch:c:hgopt}). This requires further modelling of the oracle access and considered distance measures.

\opt{jaamas}{
}

\opt{aamas}{
	\bibliographystyle{ACM-Reference-Format}  %
	\balance  %
	\bibliography{bibliography_clean}
}
\opt{arxiv}{
	\printbibliography
}
\opt{jaamas}{
	\bibliographystyle{spbasic}      %
	\bibliography{bibliography_clean}
}

\end{document}